\documentclass[aps,prd,twocolumn,superscriptaddress,nofootinbib,10pt]{revtex4-2}

\usepackage{amsmath,amssymb,mathrsfs,bm,amsthm}
\usepackage{graphicx}
\usepackage{booktabs}
\usepackage[colorlinks=true,linkcolor=blue,citecolor=blue,urlcolor=blue]{hyperref}
\usepackage{orcidlink}

\newtheorem{theorem}{Theorem}
\newtheorem{proposition}{Proposition}
\newtheorem{corollary}{Corollary}
\newtheorem{lemma}{Lemma}

\newcommand{\dd}{\mathrm{d}}
\newcommand{\ee}{\mathrm{e}}

\newcommand{\CH}{\mathcal{H}_{\rm C}}

\begin{document}

\title{Mass inflation or scalaron rigidity: Cauchy horizons in metric \texorpdfstring{$f(R)$}{f(R)} gravity}

\author{Francisco S. N. Lobo\orcidlink{0000-0002-9388-8373}}
\email{fslobo@ciencias.ulisboa.pt}
\affiliation{Instituto de Astrof\'isica e Ci\^encias do Espa\c{c}o, Faculdade de Ci\^encias da Universidade de Lisboa, Edif\'icio C8, Campo Grande, P-1749-016 Lisbon, Portugal}
\affiliation{Departamento de F\'isica, Faculdade de Ci\^encias da Universidade de Lisboa, Edif\'icio C8, Campo Grande, P-1749-016 Lisbon, Portugal}

\author{Maickol Mu\~noz-Palma\orcidlink{0009-0008-1017-4406}}
\email{m.munoz58@ufromail.cl}
\affiliation{Departamento de Ciencias F\'isicas, Universidad de La Frontera, Casilla 54-D, 4811186 Temuco, Chile}

\author{Francisco Tello-Ortiz\orcidlink{0000-0002-7104-5746}}
\email{francisco.tello@ufrontera.cl}
\affiliation{Departamento de Ciencias F\'isicas, Universidad de La Frontera, Casilla 54-D, 4811186 Temuco, Chile}

\date{}

\begin{abstract}
We develop a double-null framework for Cauchy-horizon dynamics in metric
$f(R)$ gravity that separates ordinary mass inflation from scalaron
cancellation and genuine escape from the blueshift instability. Starting from
the exact spherical Hawking-mass transport equation, we show that a regular
nondegenerate Cauchy horizon undergoes mass inflation whenever the effective
longitudinal source is eventually nonnegative, uniformly dominates the mixed
channel, and has a divergent blueshift-weighted integral. Hence bounded mass
requires sufficiently strong blueshift-integrable suppression of this source
or an independent competing channel. We then analyze the weaker branch in
which the scalaron cancels only the leading Price-tail contribution. Under
non-cancelling transverse and trace asymptotics, this forces
$|R|\to\infty$ and an asymptotically linear high-curvature theory,
$f(R)/R\to F_->0$, where $F_-\equiv\lim_{v\to\infty}f_R$. The curvature
coefficient is fixed explicitly by the transverse asymptotics. For regular
model classes with a finite high-curvature $f_{RR}$ limit one obtains
$f_{RR}\to0$. More sharply, every eventually viable branch with $f_{RR}>0$
is forced to $R\to-\infty$ and $\Lambda_\infty>-1$; the standard scalaron
mass parameter then diverges, while an additional differentiable rate
condition yields $m_{\rm sc}^2\sim F/(3f_{RR})\to+\infty$. In the Einstein
frame the scalaron null kinetic contribution is nonnegative, so the
Jordan-frame cancellation cannot be interpreted as negative scalaron null
energy. The result is a local rigidity classification rather than a global
theorem of strong cosmic censorship.
\end{abstract}
\maketitle

\section{Introduction}
\label{sec:introduction}

The inner horizon of a charged or rotating black hole is a natural arena in
which small exterior perturbations can be converted into large interior
curvature. The basic mechanism was identified through the blueshift
instability of Reissner--Nordstr\"om and Kerr Cauchy horizons and was sharpened
by the Poisson--Israel mass-inflation scenario: counter-streaming ingoing and
outgoing radiation produces an exponential growth of an interior mass
function even when the late-time influx seen outside the black hole is
arbitrarily weak~\cite{ChandrasekharHartle1982,PoissonIsrael1989,PoissonIsrael1990,Ori1991}.
Price's law supplies the slowly decaying exterior tail that feeds the
instability in asymptotically flat collapse~\cite{Price1972}. The associated
breakdown of determinism at Cauchy horizons is closely connected with
Penrose's strong-cosmic-censorship perspective~\cite{Penrose1974}.
Rigorous work in spherical symmetry subsequently established the relevant
interior blow-up and clarified its relation to extendibility of the maximal
Cauchy development~\cite{Dafermos2005,LukOh2017I,LukOh2017II}. These results
also make clear that mass inflation and strong cosmic censorship (SCC) are
closely related but not identical statements: the regularity class across the
Cauchy horizon matters, and $C^0$ extendibility can coexist with stronger
curvature or connection blow-up.

The instability is particularly important for regular black holes, whose
nonsingular cores commonly require an inner horizon.  Generic nonextremal
inner horizons inherit the mass-inflation problem
~\cite{CarballoRubio2021}.  A distinct line of work has therefore explored
\emph{inner-degenerate} geometries with vanishing inner-horizon surface
gravity, for which the exponential amplification is absent or replaced by a
weaker law~\cite{CarballoRubio2022,Franzin2022,CarballoRubio2024,Feng2026,LiuSoranidis2026,EichhornFernandes2026,DiFilippoKolarKubiznak2025}.
These models already show that the surface gravity $\kappa_-$ is not a
technical detail: changing the local horizon geometry can remove the very
exponential clock responsible for ordinary mass inflation.

Modified gravity offers another possible route.  Metric $f(R)$ gravity
contains, in addition to the massless spin-2 field, a scalar degree of
freedom $F\equiv f_R$~\cite{Buchdahl1970,SotiriouFaraoni2010,DeFeliceTsujikawa2010,CapozzielloDeLaurentis2011,Nojiri:2010wj,Nojiri:2017ncd}.
Because second derivatives of $F$ enter the Jordan-frame field equations,
they can have either sign in a null-null projection.  It is therefore
natural to ask whether the scalaron can dynamically screen the matter flux
that drives mass inflation.  Numerical double-null studies already show
that the answer is subtle.  Hwang, Lee and Yeom found mass inflation in
charged $f(R)$ collapse while observing that higher-curvature corrections
can strongly modify particular curvature scalars~\cite{HwangLeeYeom2011}.
Guo and Joshi later found that black-hole interiors can drive $F$ toward its
high-curvature general-relativistic value or toward curvature singularities,
with the behavior depending strongly on the ultraviolet completion of the
chosen $f(R)$ model~\cite{GuoJoshi2015}.  Viable dark-energy models themselves
provide examples in which the high-curvature regime is associated with a
large scalaron mass and $f_{RR}\to0$~\cite{HuSawicki2007,Starobinsky2007,Frolov2008}.

The purpose of this paper is not to assert that every $f(R)$ theory has the
same inner-horizon singularity.  Instead, we ask the more precise question:
\emph{what must happen locally if the scalaron is to screen the
	nondegenerate Cauchy-horizon mass-inflation source?}  The answer leads to a
useful classification.

Our analysis yields three complementary results. First, the exact
Hawking-mass transport law shows that, at a regular nondegenerate Cauchy
horizon in the longitudinally dominated positive-source sector, bounded mass
requires the effective longitudinal source to be sufficiently suppressed to
overcome the exponential blueshift, unless another term in the exact transport
law competes at the same order. Second, if the leading Price-tail contribution
is cancelled by the scalaron, the resulting branch is strongly constrained:
under the stated non-cancelling transverse and trace asymptotics it is driven
toward an asymptotically linear high-curvature sector. For regular high-curvature model classes in which $f_{RR}$ admits a finite asymptotic limit, this implies $f_{RR}\to0$. Third, the cancellation itself selects the
curvature direction on viable branches: because $F_{,v}<0$ asymptotically,
an eventually stable $f_{RR}>0$ branch must satisfy $R\to-\infty$ and
$\Lambda_\infty>-1$. The standard scalaron mass parameter then diverges; with
the additional rate control specified below, the usual
$m_{\rm sc}^2\sim F/(3f_{RR})$ heavy-scalaron limit follows. The Einstein-frame
representation further shows that the Jordan-frame cancellation is not a
negative scalaron null-energy contribution. These statements provide a local
rigidity classification of possible screening mechanisms, not a universal
no-go theorem for inner horizons or SCC.

The paper is organized as follows.  Section~\ref{sec:double-null} develops
the full double-null equations and derives the exact mass-transport law.
Section~\ref{sec:criterion} establishes the local mass-inflation criterion.
Section~\ref{sec:screening} analyzes the leading-cancellation branch and proves the
scalaron-rigidity theorem.  Section~\ref{sec:einstein} gives the
Einstein-frame interpretation.  Section~\ref{sec:escape} classifies the
independent escape channels and connects them with degenerate-inner-horizon
models and exponentially decaying tails.  Section~\ref{sec:models} discusses
consequences for representative $f(R)$ model classes and existing numerical
work.  Section~\ref{sec:scc} clarifies the relation to SCC and the scope of
the result.  Technical derivations are collected in the Appendices.

\section{Metric \texorpdfstring{$f(R)$}{f(R)} gravity in full double-null form}
\label{sec:double-null}

\subsection{Field equations and notation}

We work in four spacetime dimensions and use units $G=c=1$.  The metric
$f(R)$ action is
\begin{equation}
 S=\frac{1}{16\pi}\int \dd^4x\,\sqrt{-g}\,f(R)+S_m[g,\Psi].
 \label{eq:action}
\end{equation}
Defining $ F\equiv f_R=\dd f/\dd R$,
the metric field equations are
\begin{equation}
 F R_{\mu\nu}-\frac12 f g_{\mu\nu}
 -\nabla_\mu\nabla_\nu F+g_{\mu\nu}\Box F
 =8\pi T_{\mu\nu},
 \label{eq:field}
\end{equation}
with trace
\begin{equation}
 3\Box F+FR-2f=8\pi T.
 \label{eq:trace}
\end{equation}

It is useful to write Eq.~\eqref{eq:field} as
\begin{equation}
 F G_{\mu\nu}=\mathcal P_{\mu\nu},
 \label{eq:Einstein-form}
\end{equation}
where
\begin{equation}
 \mathcal P_{\mu\nu}
 =8\pi T_{\mu\nu}+\nabla_\mu\nabla_\nu F
 -g_{\mu\nu}\Box F+\frac12(f-FR)g_{\mu\nu}.
 \label{eq:Pmunu}
\end{equation}
No ``effective-fluid'' interpretation is required below; $\mathcal P_{\mu\nu}$
is only a convenient shorthand for the exact geometric source of
$G_{\mu\nu}$.

In spherical symmetry we choose double-null coordinates
\begin{equation}
 \dd s^2=-\Omega^2(u,v)\,\dd u\,\dd v+r^2(u,v)\,\dd\Omega_2^2,
 \label{eq:metric}
\end{equation}
so $g_{uv}=g_{vu}=-\Omega^2/2$ and $g^{uv}=g^{vu}=-2/\Omega^2$.

We also write $\Omega=\ee^\sigma$.  The null-null source relevant to the
$v$ direction is particularly simple because $g_{vv}=0$:
\begin{equation}
 \mathcal P_{vv}
 =8\pi T_{vv}+F_{,vv}-2\sigma_{,v}F_{,v}.
 \label{eq:Pvv}
\end{equation}
This identity is exact even for a fully two-dimensional scalaron
$F=F(u,v)$.

The scalar d'Alembertian is
\begin{equation}
 \Box F=-\frac{4}{\Omega^2}
 \left[
 F_{,uv}+\frac{r_{,u}}{r}F_{,v}
 +\frac{r_{,v}}{r}F_{,u}
 \right].
 \label{eq:boxF}
\end{equation}
The term proportional to $r_{,u}F_{,v}$ is the area-factor contribution that
is easily lost if one treats the $(u,v)$ sector as if it were a flat
two-dimensional wave operator.  Equation~\eqref{eq:boxF} is also explicit in
the established double-null formulation of $f(R)$ collapse
~\cite{HwangLeeYeom2011}.

\subsection{Hawking mass and exact transport}

We use the geometric Hawking mass (equivalently, the Misner--Sharp mass aspect in spherical symmetry), defined by
\begin{equation}
 1-\frac{2m}{r}
 =g^{ab}\partial_a r\partial_b r
 =-\frac{4}{\Omega^2}r_{,u}r_{,v},
 \label{eq:MSmass}
\end{equation}
where $a,b\in\{u,v\}$.  This is the standard quasi-local energy in
spherical symmetry~\cite{Hayward1996}.  In electrovacuum it differs from the
usual Reissner--Nordstr\"om mass parameter by the finite Coulomb term
$Q^2/(2r)$; at a Cauchy horizon with $r\to r_->0$ finite, divergence of one
mass aspect is therefore equivalent to divergence of the other.
Differentiating Eq.~\eqref{eq:MSmass} and using the $vv$ and $uv$ components
of Eq.~\eqref{eq:Einstein-form} gives the exact transport identity
\begin{equation}
 \partial_v m
 =\frac{r^2}{F\Omega^2}
 \left(
 \mathcal P_{uv}\,r_{,v}
 -\mathcal P_{vv}\,r_{,u}
 \right).
 \label{eq:mass-transport-v}
\end{equation}
Similarly,
\begin{equation}
 \partial_u m
 =\frac{r^2}{F\Omega^2}
 \left(
 \mathcal P_{uv}\,r_{,u}
 -\mathcal P_{uu}\,r_{,v}
 \right).
 \label{eq:mass-transport-u}
\end{equation}
A direct derivation is given in Appendix~\ref{app:mass}.  Equations
\eqref{eq:mass-transport-v}--\eqref{eq:mass-transport-u} replace the
heuristic ``generalized Poisson--Israel relation'' by an exact identity.
They also make clear that a longitudinal cancellation is only one possible
way of altering the mass-inflation balance: the mixed channel can in
principle compete if it ceases to be subleading.

\section{A local criterion for mass inflation}
\label{sec:criterion}

We now work along a fixed outgoing generator $u=u_0$ approaching a future
Cauchy horizon $\CH$ as $v\to\infty$. To keep the coordinate normalization
separate from the invariant statement of nondegeneracy, we define the
advanced-time blueshift exponent $\beta_->0$ by
\begin{align}
 \Omega^2(v)&=B\,\ee^{-\beta_- v}[1+o(1)],
 & B&>0,
 \label{eq:CH-Om}\\
 \sigma_{,v}(v)&\to-\frac{\beta_-}{2},
 \label{eq:CH-sigma}\\
 r(v)&\to r_->0,
 & r_{,u}(v)&\to-A_-<0,
 \label{eq:CH-r}\\
 F(v)&\to F_->0,
 & r_{,v}(v)&\to0.
 \label{eq:CH-F}
\end{align}
The derivative condition in Eq.~\eqref{eq:CH-sigma} is stated explicitly;
it does not follow from an $o(1)$ remainder in Eq.~\eqref{eq:CH-Om} without
additional regularity. Once the normalization of the advanced coordinate is
fixed, the positive exponent $\beta_-$ is determined by the nonzero
Cauchy-horizon surface gravity; for the standard asymptotically normalized
Eddington advanced time one has $\beta_-=\kappa_-$. A constant rescaling of
$v$ rescales $\beta_-$ accordingly, while the invariant content is the
presence of a positive exponential blueshift exponent. A degenerate horizon
with $\kappa_-=0$ lies outside the exponential asymptotic class
\eqref{eq:CH-Om} and must instead be described by its appropriate
subexponential, often power-law, near-horizon asymptotics. The assumptions
above therefore describe a regular nondegenerate Jordan-frame Cauchy horizon
on a branch for which the effective Newton coupling remains finite and
positive.

The condition $r_{,u}\to-A_-<0$ is logically independent of the positivity of
$\beta_-$. It encodes a nonvanishing transverse areal-radius gradient, as
appropriate to the counter-streaming regime underlying the standard
mass-inflation mechanism; it is not implied by horizon nondegeneracy alone.

The mixed contribution will be called \emph{uniformly subleading} if there
exists a nonnegative function $\varepsilon(v)\to0$ such that, for all
sufficiently large $v$,
\begin{equation}
 \left|\mathcal P_{uv}r_{,v}\right|
 \leq \varepsilon(v)\,
 \left|\mathcal P_{vv}\right|(-r_{,u}).
 \label{eq:mixed-subleading}
\end{equation}

This uniform formulation also controls isolated zeros of $\mathcal P_{vv}$:
whenever the longitudinal term vanishes, the mixed term must vanish at the
same asymptotic order. It is the local version of the usual statement that the
blueshifted longitudinal flux controls the leading mass-inflation channel.

{It is useful to distinguish the pure general-relativistic
	limit from the genuinely two-dimensional scalaron case.  When $F\equiv1$,
	all scalaron-derivative contributions to $\mathcal P_{uv}$ and
	$\mathcal P_{vv}$ vanish, so that
	$\mathcal P_{\mu\nu}=8\pi T_{\mu\nu}$.  This simplification does not,
	by itself, make the longitudinal-dominance condition
	\eqref{eq:mixed-subleading} automatic for arbitrary matter: one still
	requires the matter-sector hierarchy
	\[
	|T_{uv}r_{,v}|
	=o\!\left(|T_{vv}|(-r_{,u})\right).
	\]
	For the standard matter models underlying the usual spherical
	mass-inflation scenario this hierarchy is often satisfied by the
	matter equations themselves---for example, the radial massless-scalar
	sector has $T_{uv}=0$, while the Maxwell mixed component is suppressed
	by the near-horizon geometry.  For a genuine two-dimensional scalaron,
	however, $\mathcal P_{uv}$ additionally contains transverse derivatives,
	including $F_{,uv}$, whose behavior is not fixed by the longitudinal
	cancellation equation~\eqref{eq:screening-ode}.  Thus
	Eq.~\eqref{eq:mixed-subleading} should be regarded in the general
	$f(R)$ problem as an independent asymptotic hypothesis on the mixed
	channel, logically separate from the longitudinal cancellation and
	analogous in status to the non-cancelling transverse assumption
	$\Lambda(v)\to\Lambda_\infty\neq-1$ introduced in
	Sec.~\ref{4B}.}

\begin{theorem}[Local mass-inflation criterion]
\label{thm:massinflation}
Assume Eqs.~\eqref{eq:CH-Om}--\eqref{eq:CH-F} and the uniform
longitudinal-dominance condition~\eqref{eq:mixed-subleading}. If
$\mathcal P_{vv}\geq0$ eventually and
\begin{equation}
 \int^{\infty}\ee^{\beta_-v}\,\mathcal P_{vv}(v)\,\dd v
 =\infty,
 \label{eq:weighted-divergence}
\end{equation}
then the Hawking mass diverges along the generator,
\begin{equation}
 m(v)\to+\infty.
 \label{eq:m-diverge}
\end{equation}
\end{theorem}

\begin{proof}
Under Eqs.~\eqref{eq:CH-Om}--\eqref{eq:mixed-subleading}, the exact transport
law~\eqref{eq:mass-transport-v} gives
\begin{equation}
 \partial_v m
 =\frac{r_-^2A_-}{F_-B}
 \ee^{\beta_-v}\mathcal P_{vv}(v)[1+o(1)].
 \label{eq:mv-asymptotic}
\end{equation}
The prefactor is positive, and Eq.~\eqref{eq:weighted-divergence} therefore
forces the integral of $\partial_vm$ to diverge.
\end{proof}

Theorem~\ref{thm:massinflation} is intentionally local. It does not assume a
specific matter model and does not identify mass inflation with a global SCC
theorem. It states exactly what nondegeneracy does: the longitudinal source
is weighted by $\Omega^{-2}\sim\ee^{\beta_-v}$.

Within the eventually nonnegative sector of Theorem~\ref{thm:massinflation},
a bounded-mass branch that retains the same nondegenerate geometry, finite
positive $F$, and longitudinal dominance must satisfy
\begin{equation}
 \int^\infty \Omega^{-2}\mathcal P_{vv}\,\dd v<\infty.
 \label{eq:screening-necessary}
\end{equation}
For a sign-robust notion that does not rely on cancellations between positive
and negative residuals, we shall call
\begin{equation}
 \int^\infty \Omega^{-2}\left|\mathcal P_{vv}\right|\,\dd v<\infty
 \label{eq:screening-absolute}
\end{equation}
\emph{blueshift-integrable longitudinal screening}. For a polynomial Price
tail this is much stronger than a mere reduction of amplitude: the residual
source must be suppressed sufficiently rapidly to beat the exponential
blueshift.

\section{\texorpdfstring{Leading longitudinal cancellation and scalaron rigidity}{Leading longitudinal cancellation and scalaron rigidity}}
\label{sec:screening}

\subsection{\texorpdfstring{Price-type tails and leading longitudinal cancellation}{Price-type tails and leading longitudinal cancellation}}

To make the asymptotics precise, let
\begin{equation}
 \tau(v)\equiv T_{vv}(u_0,v)>0
 \label{eq:tau}
\end{equation}
for sufficiently large $v$, and assume
\begin{align}
 \tau(v)&\to0,
 & \int_{v_0}^\infty\tau(v)\,\dd v&<\infty,
 \label{eq:tail-L1}\\
 \frac{\tau'(v)}{\tau(v)}&\to0,
 & \ee^{\beta_-v}\tau(v)&\to\infty.
 \label{eq:tail-slow}
\end{align}
These conditions include the usual integrable power-law fluxes
$\tau\sim v^{-p}$ with $p>1$, and more generally tails with slowly varying
multiplicative corrections. The final condition states that the tail is too
slow to beat the nondegenerate blueshift. We deliberately do not claim that
this class contains exponentially decaying de Sitter tails; those form a
separate escape channel in Sec.~\ref{sec:escape}.

We call a branch \emph{leading-order longitudinally cancelling} when
\begin{equation}
 \mathcal P_{vv}=o(\tau).
 \label{eq:screened}
\end{equation}
This condition removes the Price-tail matter term at leading order, but it is
not by itself sufficient for bounded mass: a residual that is merely a
smaller power law still fails the blueshift-integrability condition
\eqref{eq:screening-absolute}. Thus Eqs.~\eqref{eq:screened} and
\eqref{eq:screening-absolute} encode distinct notions. The former isolates
the asymptotic cancellation sector analyzed below; the latter is the strong
screening condition relevant to robust bounded-mass evolution. Using
Eq.~\eqref{eq:Pvv}, the leading-cancellation equation is
\begin{equation}
 F_{,vv}-2\sigma_{,v}F_{,v}
 =-8\pi\tau+o(\tau).
 \label{eq:screening-ode}
\end{equation}
No assumption has been made about $F_{,u}$ or $F_{,uv}$.

\begin{lemma}[Universal longitudinal rate]
\label{lem:longitudinal}
Suppose Eqs.~\eqref{eq:CH-Om} and~\eqref{eq:CH-sigma} hold and the tail obeys
Eqs.~\eqref{eq:tail-L1}--\eqref{eq:tail-slow}. Then every solution of
Eq.~\eqref{eq:screening-ode} satisfies
\begin{equation}
 F_{,v}(v)
 =-\frac{8\pi}{\beta_-}\,\tau(v)[1+o(1)].
 \label{eq:Fv-universal}
\end{equation}
Consequently $F(v)$ has a finite limit. On the regular branch considered
here we assume this limit is $F_->0$.
\end{lemma}

\begin{proof}
Multiplying Eq.~\eqref{eq:screening-ode} by $\ee^{-2\sigma}$ gives
\begin{equation}
 \partial_v\!\left(\ee^{-2\sigma}F_{,v}\right)
 =-8\pi\ee^{-2\sigma}\tau[1+o(1)].
 \label{eq:integrating-factor}
\end{equation}
The homogeneous term is proportional to $\ee^{2\sigma}$ and is
exponentially suppressed. Moreover,
\begin{equation}
 \int^{v}\ee^{-2\sigma}\tau\,\dd\bar v
 \sim\frac{\ee^{-2\sigma(v)}\tau(v)}{\beta_-},
 \label{eq:asymptotic-integral}
\end{equation}
which follows directly from l'H\^opital's rule because
$-2\sigma_{,v}+\tau'/\tau\to\beta_-$. Substitution into
Eq.~\eqref{eq:integrating-factor} yields Eq.~\eqref{eq:Fv-universal}.
Finally, Eq.~\eqref{eq:tail-L1} implies $F_{,v}\in L^1$ and hence a finite
limit for $F$.
\end{proof}

Lemma~\ref{lem:longitudinal} already has a useful interpretation.  Leading-order cancellation
is not achieved by keeping the scalaron dynamically large. Along the
leading-cancellation branch its longitudinal derivative is slaved to the decaying matter
flux and tends to zero.  The remaining question is whether this longitudinal
behavior is compatible with the transverse scalaron equation encoded in the
trace.

\subsection{The transverse channel in \texorpdfstring{$\Box F$}{Box F}}\label{4B}

Introduce the dimensionless transverse ratio
\begin{equation}
 \Lambda(v)
 \equiv
 \frac{F_{,uv}+(r_{,v}/r)F_{,u}}
 {(r_{,u}/r)F_{,v}},
 \label{eq:Lambda-def}
\end{equation}
whenever the denominator is nonzero.  Equation~\eqref{eq:boxF} becomes
\begin{equation}
 \Box F
 =-\frac{4}{\Omega^2}\frac{r_{,u}}{r}F_{,v}
 \left[1+\Lambda(v)\right].
 \label{eq:box-Lambda}
\end{equation}
The exceptional value $\Lambda\to-1$ corresponds to a distinct
\emph{transverse cancellation}: $F_{,uv}+(r_{,v}/r)F_{,u}$ cancels the
area-factor term.  Nothing in the longitudinal leading-cancellation equation enforces
this cancellation.

We shall say that the transverse scalaron dynamics is \emph{asymptotically non-cancelling} when
\begin{equation}
 \Lambda(v)\to\Lambda_\infty,
 \qquad
 1+\Lambda_\infty\neq0.
 \label{eq:transverse-generic}
\end{equation}

Then Lemma~\ref{lem:longitudinal} and
Eqs.~\eqref{eq:CH-Om}--\eqref{eq:CH-r} imply
\begin{equation}
 \Box F
 \sim
 -\frac{32\pi A_-}{\beta_-Br_-}
 (1+\Lambda_\infty)
 \ee^{\beta_-v}\tau(v).
 \label{eq:box-asymptotic}
\end{equation}

Thus the scalaron contribution removed from the longitudinal source reappears in the trace equation with the same blueshift factor under the non-cancelling transverse asymptotics~\eqref{eq:transverse-generic}.

\subsection{Trace closure and the rigidity theorem}

Let $T\equiv T^\mu{}_{\mu}$ be the matter trace.  We assume that it does not
cancel the blueshifted scalaron term at leading order,
\begin{equation}
 T(v)=o\!\left(\ee^{\beta_-v}\tau(v)\right).
 \label{eq:trace-subleading}
\end{equation}

This includes trace-free matter such as Maxwell fields and null dust and is
sufficient, though not necessary, for the theorem below.  A leading-order
trace cancellation will be listed separately as an escape channel.

\begin{theorem}[Asymptotic scalaron rigidity under leading longitudinal cancellation]
\label{thm:rigidity}
Assume Eqs.~\eqref{eq:CH-Om}--\eqref{eq:CH-F},
\eqref{eq:tail-L1}--\eqref{eq:tail-slow}, leading-order longitudinal
cancellation~\eqref{eq:screened}, non-cancelling transverse asymptotics
\eqref{eq:transverse-generic}, and trace regularity
\eqref{eq:trace-subleading}. Assume also that $f(R)$ is $C^2$ on an open
interval containing the relevant curvature branch and, if the sampled
curvature range remains bounded, that the closure of that range is contained
in this regular interval. Assume further that $F=f_R\to F_->0$, and that
$R(v)$ is eventually monotone if it becomes unbounded. Then
\begin{equation}
 |R(v)|\to\infty,
 \qquad
 R(v)=C_R\ee^{\beta_-v}\tau(v)[1+o(1)],
 \label{eq:R-asymptotic}
\end{equation}
where the nonzero coefficient is fixed by
\begin{equation}
 C_R=-\frac{96\pi A_-}{\beta_-Br_-F_-}
 \left(1+\Lambda_\infty\right) ,
 \label{eq:CR-explicit}
\end{equation}
and
\begin{equation}
 \frac{f(R(v))}{R(v)}\to F_- .
 \label{eq:R-and-linear}
\end{equation}
Thus the leading-cancellation high-curvature branch is asymptotically linear,
\begin{equation}
 f(R)=F_-R+o(R),
 \label{eq:asymptotic-linear}
\end{equation}
along the Cauchy-horizon trajectory.
\end{theorem}

{The theorem should therefore be read as a conditional
	rigidity statement rather than an existence result: it does not assert
	that the leading-cancellation branch is dynamically realized, but shows
	that any regular branch realizing such a cancellation under the stated
	asymptotics is strongly constrained in its high-curvature endpoint.}

\begin{proof}
By Eq.~\eqref{eq:box-asymptotic} and
$\ee^{\beta_-v}\tau(v)\to\infty$, one has $|\Box F|\to\infty$. If $R$
remained bounded, the closure of its sampled range would lie in the regular
$C^2$ interval specified in the theorem, so $FR-2f$ would remain bounded,
contradicting the trace equation~\eqref{eq:trace} together with
Eq.~\eqref{eq:trace-subleading}. Hence $R$ is unbounded. The assumed eventual
monotonicity then gives $|R|\to\infty$.

Because $F(R(v))\to F_-$ along an eventually monotone unbounded curvature
branch, the elementary Ces\`aro/l'H\^opital limit gives
$f(R(v))/R(v)\to F_-$. Therefore
\begin{equation}
 FR-2f=-F_-R[1+o(1)].
 \label{eq:algebraic-trace}
\end{equation}
Substituting Eq.~\eqref{eq:box-asymptotic} and
Eq.~\eqref{eq:algebraic-trace} into the trace equation, while using
Eq.~\eqref{eq:trace-subleading}, yields
\begin{equation}
 F_-R
 =-\frac{96\pi A_-}{\beta_-Br_-}
 (1+\Lambda_\infty)\ee^{\beta_-v}\tau(v)[1+o(1)],
\end{equation}
which proves Eqs.~\eqref{eq:R-asymptotic} and~\eqref{eq:CR-explicit}.
Equation~\eqref{eq:R-and-linear} is equivalent to
Eq.~\eqref{eq:asymptotic-linear} along the branch.
\end{proof}

\begin{corollary}[Curvature-sign selection on an eventually viable branch]
\label{cor:viable-sign}
Under the hypotheses of Theorem~\ref{thm:rigidity}, suppose in addition that
$f_{RR}(R(v))>0$ for all sufficiently large $v$. Then
\begin{equation}
 R_{,v}<0,
 \qquad
 R(v)\to-\infty,
 \qquad
 \Lambda_\infty>-1.
 \label{eq:viable-sign}
\end{equation}
Conversely, if $\Lambda_\infty<-1$, then $R(v)\to+\infty$ and the
leading-cancellation branch necessarily has $f_{RR}<0$ eventually wherever
$R_{,v}\neq0$.
\end{corollary}

\begin{proof}
Lemma~\ref{lem:longitudinal} gives $F_{,v}<0$ eventually because
$\tau>0$ and $\beta_->0$. Since $F=f_R(R)$,
\begin{equation}
 F_{,v}=f_{RR}R_{,v}.
 \label{eq:Fv-chain}
\end{equation}
Thus $f_{RR}>0$ implies $R_{,v}<0$. Together with
$|R|\to\infty$ this forces $R\to-\infty$. Since
$\ee^{\beta_-v}\tau(v)>0$, Eq.~\eqref{eq:R-asymptotic} then requires
$C_R<0$, and Eq.~\eqref{eq:CR-explicit} gives
$1+\Lambda_\infty>0$. If instead $\Lambda_\infty<-1$, then
$C_R>0$ and hence $R\to+\infty$; eventual monotonicity gives
$R_{,v}>0$ wherever nonzero, while $F_{,v}<0$, so
$f_{RR}=F_{,v}/R_{,v}<0$.
\end{proof}

\begin{corollary}[Asymptotic vanishing of $f_{RR}$ for regular high-curvature model classes]
\label{cor:decoupling}
Under the hypotheses of Theorem~\ref{thm:rigidity}, suppose in addition that
$f_{RR}(R)$ has a finite well-defined limit as $|R|\to\infty$ along the relevant
branch.  Then
\begin{equation}
 f_{RR}(R(v))\to0.
 \label{eq:fRRzero}
\end{equation}
\end{corollary}

\begin{proof}
Since $F=f_R$ approaches the finite value $F_-$ while $|R|\to\infty$, a
nonzero finite limit of $f_{RR}$ would force $F$ to grow linearly with $R$,
contradicting $F\to F_-$.  Hence the only possible finite asymptotic limit is
zero.
\end{proof}

\begin{corollary}[Divergence of the standard scalaron mass parameter]
\label{cor:mass-divergence}
Under the hypotheses of Theorem~\ref{thm:rigidity}, if $F>0$ and
$f_{RR}>0$ eventually, then the standard Jordan-frame scalaron mass parameter
\begin{equation}
 m_{\rm sc}^2=\frac13\left(\frac{F}{f_{RR}}-R\right)
 \label{eq:mass-divergence-def}
\end{equation}
satisfies
\begin{equation}
 m_{\rm sc}^2\to+\infty.
 \label{eq:mass-divergence}
\end{equation}
\end{corollary}

\begin{proof}
Corollary~\ref{cor:viable-sign} gives $R\to-\infty$. Hence
\begin{equation}
 m_{\rm sc}^2
 =\frac13\left(\frac{F}{f_{RR}}+|R|\right)
 >\frac{|R|}{3}\to+\infty.
\end{equation}
\end{proof}

\begin{proposition}[Rate-controlled heavy-scalaron limit]
\label{prop:fRRrate}
Under the hypotheses of Theorem~\ref{thm:rigidity}, suppose the asymptotic
closure is differentiable in the sense that
\begin{equation}
 \frac{R_{,v}}{R}\to\beta_-.
 \label{eq:RvR}
\end{equation}
Then
\begin{align}
 f_{RR}(R(v))&=O\!\left(\ee^{-\beta_-v}\right),
 \label{eq:fRR-rate}\\
 R(v)f_{RR}(R(v))&=O\!\left(\tau(v)\right)\to0.
 \label{eq:RfRR-rate}
\end{align}
If, in addition, $f_{RR}(R(v))>0$ eventually, then
\begin{equation}
 m_{\rm sc}^2\sim\frac{F}{3f_{RR}}\longrightarrow+\infty.
 \label{eq:mass-decoupling}
\end{equation}
\end{proposition}

\begin{proof}
Where $R_{,v}\neq0$, the chain rule gives
\begin{equation}
 f_{RR}=\frac{F_{,v}}{R_{,v}}.
 \label{eq:chainrule}
\end{equation}
Combining Eq.~\eqref{eq:Fv-universal},
Eq.~\eqref{eq:R-asymptotic}, and Eq.~\eqref{eq:RvR} yields
Eq.~\eqref{eq:fRR-rate}. Multiplying this estimate by
Eq.~\eqref{eq:R-asymptotic} gives Eq.~\eqref{eq:RfRR-rate}, because
$\tau(v)\to0$. Thus $Rf_{RR}/F\to0$. If $f_{RR}>0$ eventually, then
Eq.~\eqref{eq:fRR-rate} gives $f_{RR}\to0^+$ and
\begin{equation}
 m_{\rm sc}^2
 =\frac{F}{3f_{RR}}
 \left(1-\frac{Rf_{RR}}{F}\right)
 \sim\frac{F}{3f_{RR}}\to+\infty,
\end{equation}
which proves Eq.~\eqref{eq:mass-decoupling}.
\end{proof}

The asymptotic hierarchy underlying the two principal branches is summarized
schematically in Fig.~\ref{fig:asymptotic-hierarchy}.  Panel (a) illustrates
the standard mass-inflation regime, in which a slowly decaying longitudinal
source remains nonintegrable after multiplication by the Cauchy-horizon
blueshift factor.  Panel (b) shows the distinct leading-cancellation branch:
although $\mathcal P_{vv}=o(\tau)$ removes the leading Price-tail source,
$F_{,v}$ remains slaved to $\tau$, while the non-cancelling trace sector drives
$\Box F$ and $|R|$ with the blueshifted scaling derived above.  The figure also
emphasizes that leading cancellation is weaker than the absolute weighted
integrability condition~\eqref{eq:screening-absolute} required for robust
bounded-mass screening.

\begin{figure*}[t]
    \centering
    \includegraphics[width=0.9\textwidth]{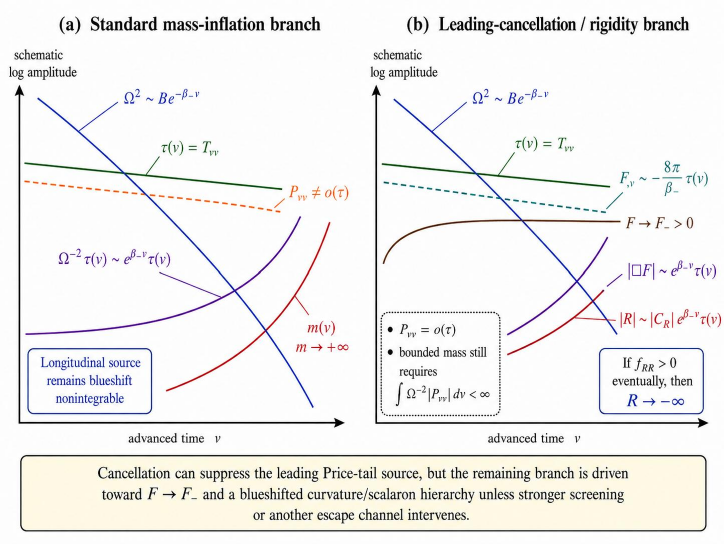}
    \caption{Schematic asymptotic scaling of the two principal local branches
    near a regular nondegenerate Cauchy horizon.  Panel (a) represents the
    standard mass-inflation regime: a Price-type tail decays slowly, whereas
    the blueshift factor $\Omega^{-2}$ grows exponentially, so the weighted
    longitudinal source remains nonintegrable and the Hawking mass diverges.
    Panel (b) represents the leading-cancellation branch: the condition
    $\mathcal P_{vv}=o(\tau)$ fixes
    $F_{,v}\sim-(8\pi/\beta_-)\tau$ and hence $F\to F_->0$, while, in the
    absence of independent transverse or trace cancellation, $\Box F$ and
    $|R|$ inherit the blueshifted scaling $\ee^{\beta_-v}\tau(v)$.  On an
    eventually viable branch with $f_{RR}>0$, the curvature is sign-selected
    to $R\to-\infty$.  The curves indicate asymptotic orders only and are not
    numerical solutions.}
    \label{fig:asymptotic-hierarchy}
\end{figure*}

Theorem~\ref{thm:rigidity} is the central result of this paper. Its content is
deliberately narrower, but stronger, than the statement that a small
$f_{RR}$ is ``pathological.'' If a regular nondegenerate branch cancels the
leading longitudinal Price-tail source and no independent transverse or trace
cancellation is arranged, then the high-curvature theory is forced toward an
asymptotically linear form, with the sign of the curvature fixed by the
transverse coefficient through Eq.~\eqref{eq:CR-explicit}. The condition
$\mathcal P_{vv}=o(\tau)$ should not be confused with the stronger
blueshift-integrable screening condition~\eqref{eq:screening-absolute}; the
theorem characterizes the asymptotic rigidity of the leading-cancellation
sector irrespective of whether the residual is already small enough to keep
the Hawking mass bounded.

For standard smooth model classes in which $f_{RR}$ has a definite finite
high-curvature limit, Corollary~\ref{cor:decoupling} gives $f_{RR}\to0$.
Corollary~\ref{cor:viable-sign} adds a separate sign statement: on an
eventually viable $f_{RR}>0$ branch, the cancellation trajectory must run to
$R\to-\infty$ and must satisfy $\Lambda_\infty>-1$. The standard algebraic
mass parameter then diverges already by Corollary~\ref{cor:mass-divergence}.
The additional differentiable asymptotics of Proposition~\ref{prop:fRRrate}
are needed for the stronger dominance relation $Rf_{RR}/F\to0$ and the
conventional $m_{\rm sc}^2\sim F/(3f_{RR})$ heavy-scalaron decoupling
asymptotics. None of these statements is, by itself, a contradiction
~\cite{SotiriouFaraoni2010,DeFeliceTsujikawa2010}.

\begin{corollary}[Obstruction for nondecoupling high-curvature branches]
\label{cor:nondecoupling}
Under the hypotheses of Theorem~\ref{thm:rigidity}, an $f(R)$ model for which
$|f_{RR}|$ is bounded away from zero along the relevant high-curvature branch
cannot support the leading-cancellation asymptotics
\eqref{eq:screened} without activating an additional channel outside the
hypotheses of the theorem.
\end{corollary}

{There is a useful physical interpretation of the rigidity
	result.  The leading-cancellation branch is driven toward an
	asymptotically linear high-curvature sector with $F\to F_->0$.  For the
	regular high-curvature model classes covered by
	Corollary~\ref{cor:decoupling}, this additionally implies
	$f_{RR}\to0$, while on an eventually viable branch with $F>0$ and
	$f_{RR}>0$ the standard scalaron mass parameter diverges.  Under the
	additional rate condition of Proposition~\ref{prop:fRRrate}, this becomes
	the conventional heavy-scalaron limit
	$m_{\rm sc}^2\sim F/(3f_{RR})\to+\infty$.  In this restricted but
	physically relevant sense, leading scalaron cancellation is
	self-restricting in theory space: rather than leaving an arbitrary
	high-curvature screening sector, it drives the sampled trajectory toward
	an increasingly Einstein-like asymptotic form.  This interpretation is a
	consequence of the stated regularity and rate assumptions and should not
	be read as an independent proof of dynamical decoupling.}

\section{Einstein-frame interpretation: why screening is not negative scalaron energy}
\label{sec:einstein}

The scalaron-rigidity result becomes physically clearer after resolving the
extra degree of freedom explicitly.  On a branch with $F>0$, define
\begin{equation}
 \widetilde g_{\mu\nu}=F g_{\mu\nu},
 \qquad
 \phi=\sqrt{\frac{3}{16\pi}}\ln F.
 \label{eq:conformal}
\end{equation}
Where the scalar-tensor representation is regular, the action is equivalent
to Einstein gravity coupled to a canonical scalar with potential
\begin{equation}
 U(\phi)=\frac{RF-f}{16\pi F^2},
 \label{eq:potential}
\end{equation}
plus matter nonminimally coupled through $F^{-1}\widetilde g_{\mu\nu}$
~\cite{SotiriouFaraoni2010,DeFeliceTsujikawa2010,FaraoniGunzigNardone1999}.
The Einstein-frame equations are
\begin{equation}
 \widetilde G_{\mu\nu}
 =8\pi\left(
 \widetilde T^{(m)}_{\mu\nu}
 +T^{(\phi)}_{\mu\nu}
 \right),
 \label{eq:EF-Einstein}
\end{equation}
with
\begin{equation}
 T^{(\phi)}_{\mu\nu}
 =\widetilde\nabla_\mu\phi\widetilde\nabla_\nu\phi
 -\frac12\widetilde g_{\mu\nu}(\widetilde\nabla\phi)^2
 -U\widetilde g_{\mu\nu}.
 \label{eq:phi-stress}
\end{equation}
For covariant matter components,
\begin{equation}
 \widetilde T^{(m)}_{\mu\nu}=F^{-1}T^{(m)}_{\mu\nu}.
 \label{eq:matter-transform}
\end{equation}
Hence, for any null vector $k^\mu$,
\begin{equation}
 \widetilde T^{\rm tot}_{kk}
 =F^{-1}T_{kk}+(k\phi)^2.
 \label{eq:null-positive}
\end{equation}
The scalar potential does not appear because $k^2=0$.  Therefore, if the physical matter satisfies $T_{kk}\ge0$, the scalaron contribution to the total Einstein-frame null flux is nonnegative. In particular, the scalaron cannot supply a negative null-kinetic contribution that cancels a positive matter null flux.

Equation~\eqref{eq:null-positive} is not used as a choice of ``physical
frame.''  It is a field-redefinition diagnostic: it shows that a cancellation
in the Jordan combination $\mathcal P_{vv}$ is not a frame-invariant
cancellation of physical null energy.  The second derivative
$\nabla_v\nabla_vF$ that can oppose $T_{vv}$ in the Jordan equation is
reorganized into a canonical positive kinetic term plus conformal geometry
in the Einstein representation.

For the scalar-curvature mode linearized about a constant-curvature, or
sufficiently slowly varying, background, the standard Jordan-frame mass
parameter is
\begin{equation}
 m_{\rm sc}^2
 =\frac13\left(\frac{F}{f_{RR}}-R\right),
 \label{eq:scalaron-mass}
\end{equation}
with the precise local interpretation depending on the background and
stability conditions~\cite{DeFeliceTsujikawa2010}. The sign selection found
above makes the hierarchy transparent. On an eventually viable branch with
$F>0$ and $f_{RR}>0$, Corollary~\ref{cor:viable-sign} gives $R\to-\infty$;
therefore the algebraic parameter~\eqref{eq:scalaron-mass} diverges positively
even before a detailed rate comparison is imposed. Proposition~\ref{prop:fRRrate}
then supplies the stronger result $Rf_{RR}/F\to0$, so that
$m_{\rm sc}^2\sim F/(3f_{RR})\to+\infty$. It is this rate-controlled regime
that most directly supports the standard large-mass, short-range scalaron
decoupling interpretation. Such behavior should not be confused with a proof
of strong coupling or inconsistency.

\section{Classification of escape channels}
\label{sec:escape}

The preceding results organize the robust ways in which the standard
nondegenerate mass-inflation channel can be avoided within spherical symmetry
and the local assumptions used above. Table~\ref{tab:escape} lists the main
geometric, dynamical, and regularity-changing mechanisms. A conditionally
convergent sign-changing weighted source is mathematically possible as well,
but it is not a sign-robust screening mechanism and is not listed separately.

\begin{table*}[t]
\caption{Local escape channels from the standard nondegenerate mass-inflation
criterion.  The table separates geometric changes from dynamical screening
and from failures of the regular-branch assumptions.}
\label{tab:escape}
\begin{ruledtabular}
\begin{tabular}{p{0.20\textwidth}p{0.31\textwidth}p{0.39\textwidth}}
Channel & Mathematical condition & Physical interpretation \\
\hline
Horizon degeneracy & $\kappa_-=0$ & Lies outside the positive-exponential class~\eqref{eq:CH-Om}; amplification can become power-law or finite. \\
Fast exterior decay & $\int \Omega^{-2}\mathcal P_{vv}\,\dd v<\infty$ without scalaron cancellation & Relevant for sufficiently rapid, e.g. exponential, tails; central in de Sitter SCC analyses. \\
Blueshift-integrable scalaron screening & $\int\Omega^{-2}|\mathcal P_{vv}|\,\dd v<\infty$ & Removes the longitudinal divergence robustly; when realized through a smooth leading cancellation of a Price tail, the associated cancellation sector is constrained by Theorem~\ref{thm:rigidity}. \\
Transverse cancellation & $\Lambda\to-1$ & Cancels the area-factor term inside $\Box F$; independent tuning not implied by longitudinal cancellation. \\
Trace cancellation & $8\pi T$ cancels $3\Box F$ at leading order & Prevents the trace equation from forcing the rigidity curvature scaling. \\
Mixed-channel compensation & $\mathcal P_{uv}r_{,v}$ competes with $-\mathcal P_{vv}r_{,u}$ & Invalidates the usual longitudinal dominance in the exact Hawking-mass transport law. \\
Nonregular scalar-tensor branch & $F\to0$, $F\to\infty$, or loss of regularity/invertibility & The regular Einstein-frame/scalaron description ceases to be uniformly applicable. \\
Null-positivity violation & $T_{vv}<0$ over the relevant asymptotic regime & Requires exotic/quantum matter and lies outside the classical positive-flux setting. \\
\end{tabular}
\end{ruledtabular}
\end{table*}

\subsection{Degenerate inner horizons}

The cleanest geometric escape is $\kappa_-=0$. Such a degenerate inner
horizon lies outside the positive-exponential asymptotic class
\eqref{eq:CH-Om}; one must replace that ansatz by the appropriate
subexponential near-horizon behavior. Theorem~\ref{thm:massinflation} therefore
does not apply directly. This mechanism has been implemented explicitly in
spherical and rotating regular black-hole models
~\cite{CarballoRubio2022,Franzin2022}. More recent analyses show that
degenerate inner horizons can soften the standard exponential amplification,
with shell and Ori-type constructions exhibiting power-law behavior and, in
some models, finite late-time Misner--Sharp mass
~\cite{Feng2026,LiuSoranidis2026}. Modified-gravity constructions with
inner-extremal horizons provide further examples
~\cite{EichhornFernandes2026,DiFilippoKolarKubiznak2025}.
A complementary quasitopological-gravity null-shell analysis found that significant amplification, within the adopted double-shell model, is confined to parametrically small radial separations from the inner horizon~\cite{FrolovZelnikov2026}. Subsequent work has emphasized an important qualification: for the pure-gravity regular black-hole solutions considered in resummed quasitopological gravity, standard distributional null thin shells need not exist, so the corresponding thin-shell mass-inflation argument cannot be applied directly and the full inner-horizon stability problem remains open~\cite{DiFilippoKubiznakSrinivasan2026}.

This literature is important for the interpretation of the present result.

Our rigidity theorem does not say that modified gravity has no route to a
stable interior. It says that a \emph{scalaron attempt to cancel the leading
longitudinal source at fixed nondegenerate geometry} is highly constrained;
actual bounded-mass screening additionally requires the weighted
integrability condition~\eqref{eq:screening-absolute}. Changing the geometry to $\kappa_-=0$, thereby leaving the positive-exponential
asymptotic class used in the theorem, is a genuinely different mechanism.

\subsection{Fast tails and the de Sitter contrast}

The second clean escape is kinematic: the incoming flux can decay faster
than the blueshift amplifies it. If, in the eventually nonnegative
longitudinal sector,
\begin{equation}
 \int^\infty \ee^{\beta_-v}\mathcal P_{vv}(v)\,\dd v<\infty,
 \label{eq:fast-tail}
\end{equation}
then the longitudinal contribution to the mass need not diverge. More
robustly, absolute weighted integrability is precisely
Eq.~\eqref{eq:screening-absolute}. This possibility is not academic. For
Einstein--Maxwell--scalar dynamics with a positive cosmological constant and
an exponential Price law, rigorous work shows that mass inflation may or may
not occur depending on the competition between decay and Cauchy-horizon
blueshift~\cite{CostaEtAl2018}. The subexponential assumptions of
Sec.~\ref{sec:screening} are therefore intended for asymptotically flat Price
tails and should not be exported unchanged to de Sitter settings.

\subsection{Dynamical and matter-sector escape channels}

The remaining channels are dynamical.  A transverse cancellation
$\Lambda\to-1$ removes the leading $\Box F$ term and evades
Theorem~\ref{thm:rigidity} without changing the leading longitudinal cancellation
condition.  A leading cancellation between $3\Box F$ and $8\pi T$ in the
trace equation provides another possibility.  Neither follows from the
$vv$ equation; both constitute additional constraints on the two-dimensional
solution.

Likewise, Eq.~\eqref{eq:mass-transport-v} shows exactly where the usual
Poisson--Israel hierarchy can fail: if $\mathcal P_{uv}r_{,v}$ remains
comparable to $-\mathcal P_{vv}r_{,u}$, the mixed channel must be retained.
Finally, violation of the classical matter null condition or a singular
limit of $F$ can modify the argument at a more fundamental level.

The classification also connects naturally with the broader quasi-local
view of mass inflation.  Exponential energy buildup can occur near slowly
evolving inner trapping horizons even when there is no exact global Cauchy
horizon~\cite{CarballoRubio2024}.  In such situations the local transport
identity remains useful, but the asymptotic limit $v\to\infty$ should be
replaced by the finite-duration evolution of the inner trapping horizon.

The complete local classification is summarized in
Fig.~\ref{fig:logic-overview}.  Under uniform longitudinal dominance, an
eventually nonnegative and blueshift-nonintegrable source yields mass
inflation.  Leading Price-tail cancellation instead enters the scalaron-rigidity
branch unless an independent transverse or trace cancellation intervenes;
horizon degeneracy, sufficiently rapid exterior decay, and mixed-channel
competition are distinct escape mechanisms.  Figure~\ref{fig:logic-overview}
is therefore a synthesis of the results derived above, not an additional
dynamical assumption.

\begin{figure*}[!t]
    \centering
    \includegraphics[width=0.98\textwidth]{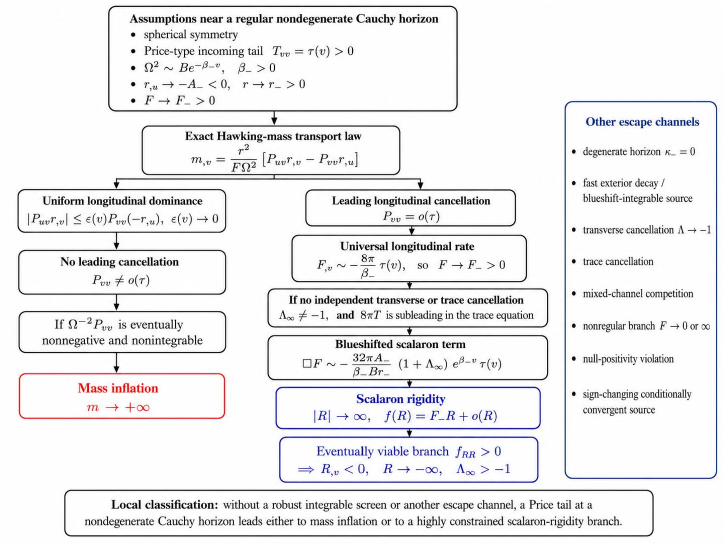}
    \caption{Logical summary of the local Cauchy-horizon classification
    developed in this work.  At a regular nondegenerate horizon, an eventually
    nonnegative and blueshift-nonintegrable longitudinal source leads to mass
    inflation when it uniformly dominates the mixed channel.  Cancelling the
    leading Price-tail contribution instead fixes the asymptotic scalaron rate
    and, in the absence of independent transverse or trace cancellations,
    drives the solution toward the scalaron-rigidity branch of
    Theorem~\ref{thm:rigidity}.  On an eventually viable branch with
    $f_{RR}>0$, the curvature is further selected toward $R\to-\infty$.  The
    right-hand column lists logically distinct escape channels that invalidate
    one or more hypotheses of this hierarchy.  The diagram is a schematic
    summary of the analytical results and does not represent a numerical
    evolution.}
    \label{fig:logic-overview}
\end{figure*}

\section{Representative \texorpdfstring{$f(R)$}{f(R)} models and numerical collapse}
\label{sec:models}

Figure~\ref{fig:logic-overview} provides a theory-independent summary of the
local alternatives.  We now ask how representative metric $f(R)$ models
populate this classification.  The rigidity theorem and its corollaries distinguish model classes by both
the magnitude and the sign structure of their high-curvature behavior. In
particular, eventual viability $f_{RR}>0$ selects $R\to-\infty$ on the
leading-cancellation branch, while the rate-controlled heavy-scalaron
interpretation uses Proposition~\ref{prop:fRRrate}.

\paragraph{Quadratic corrections.---}
For
\begin{equation}
 f(R)=R+\alpha R^2,
 \label{eq:R2}
\end{equation}
one has $f_{RR}=2\alpha$.  Corollary~\ref{cor:nondecoupling} therefore rules
out the regular leading-cancellation asymptotics of Theorem~\ref{thm:rigidity} for fixed
$\alpha\neq0$ unless another escape channel is activated.  This is notable
because quadratic models have been used extensively as ultraviolet
completions and were among the theories studied in numerical $f(R)$ collapse
~\cite{HwangLeeYeom2011,GuoJoshi2015}.

\paragraph{Viable dark-energy models.---}
Hu--Sawicki and Starobinsky-type late-time acceleration models are designed
to approach general relativity at large positive cosmological curvature, with
$F$ near a constant and a heavy scalaron
~\cite{HuSawicki2007,Starobinsky2007}. Their familiar GR-recovery behavior is
therefore qualitatively reminiscent of the endpoint $F\to F_-$ and
$f_{RR}\to0$. However, Corollary~\ref{cor:viable-sign} shows that endpoint
resemblance is not sufficient: an eventually stable leading-cancellation
branch with $f_{RR}>0$ is forced toward $R\to-\infty$. Compatibility with a
specific dark-energy model must therefore be assessed on its negative-curvature
branch; it does not follow automatically from the standard positive-curvature
cosmological GR-recovery limit. Frolov showed that viable dark-energy models can
place a curvature singularity at a finite scalaron field distance
~\cite{Frolov2008}, and black-hole collapse can drive the scalaron toward
high-curvature regions~\cite{GuoJoshi2015}.

\paragraph{Relation to existing double-null simulations.---}
The numerical work of Hwang, Lee and Yeom is especially relevant because it
uses a full double-null $f(R)$ system and finds mass inflation near the
Cauchy horizon while higher-curvature terms can keep the Ricci scalar bounded
for particular models~\cite{HwangLeeYeom2011}.  This does not contradict
Theorem~\ref{thm:rigidity}.  Their simulations do not impose the leading-order longitudinal cancellation
condition~\eqref{eq:screened}; mass inflation remains
active.  Our result instead characterizes the special branch that would be
required to screen that source.  Conversely, the fact that a particular
curvature scalar can remain bounded during mass inflation illustrates why
one should distinguish Hawking-mass blow-up, curvature blow-up, and SCC
regularity statements.

The later simulations of Guo and Joshi found that, in dark-energy $f(R)$
models, strong interior dynamics can push $F$ toward its GR value while
$R$ becomes singular; adding an $R^2$ term changes the location and character
of the singularity but does not generically remove strong-curvature collapse
~\cite{GuoJoshi2015}.  The scalaron-rigidity theorem provides an analytic organizing principle for
these observations: a branch that tries to cancel the longitudinal source
while remaining regular is driven toward an asymptotically linear scalaron
sector. On an eventually stable $f_{RR}>0$ branch, the cancellation trajectory
is additionally forced toward negative divergent curvature. Under the rate
condition of Proposition~\ref{prop:fRRrate}, the standard $F/f_{RR}$ term
dominates the scalaron mass parameter and yields the conventional heavy,
short-range limit. A branch that does not satisfy this hierarchy must instead
activate another escape channel.

\section{Mass inflation, strong cosmic censorship, and scope}
\label{sec:scc}

The present results concern a local instability mechanism and should not be
promoted directly into a theorem of strong cosmic censorship.  In the
spherically symmetric Einstein--Maxwell--scalar system, the rigorous picture
is already subtle: the metric can remain $C^0$ extendible across a Cauchy
horizon while stronger curvature or connection quantities diverge, and the
$C^2$ formulation of SCC requires generic lower bounds on the radiation
field~\cite{Dafermos2005,LukOh2017I,LukOh2017II}.  With a positive
cosmological constant, exponentially decaying tails can alter mass inflation
and the corresponding extendibility properties~\cite{CostaEtAl2018}.

Accordingly, our statements have the following scope.

First, Theorem~\ref{thm:massinflation} is an exact local criterion in
spherical symmetry.  It requires a nondegenerate Cauchy-horizon asymptotic
regime and longitudinal dominance of the Hawking-mass transport equation.
It does not prove that such a horizon forms from generic Cauchy data. The
assumption $r_{,u}\to-A_-<0$ is an independent counter-streaming condition and
is not implied by nondegeneracy alone.

Second, Theorem~\ref{thm:rigidity} is conditional on leading-order
longitudinal cancellation, Eq.~\eqref{eq:screened}, not on the stronger
bounded-mass condition~\eqref{eq:screening-absolute}. It classifies the
asymptotics of a branch that removes the leading $vv$ Price-tail source. It
does not claim that the residual source is already blueshift integrable, that
such a branch exists globally, that it is dynamically selected, or that it
is stable.

Third, the non-cancelling transverse and trace asymptotic assumptions are explicit.  If the
full two-dimensional scalaron dynamics enforces $\Lambda\to-1$, or if the
matter trace cancels the blueshifted $\Box F$ term, the rigidity theorem must
be replaced by a separate analysis.  Such a relation would itself be a
nontrivial dynamical constraint and is not hidden inside the present proof.

Fourth, the viable-branch sign corollary assumes eventual $f_{RR}>0$. It
therefore constrains only the stable sign sector of the cancellation branch;
a model may evade that sign selection by entering $f_{RR}<0$, but then it has
left the conventional viable branch.

Fifth, the Einstein-frame discussion is used only to resolve the sign of
the scalaron's kinetic null flux.  The analysis does not assume that the
Einstein frame is the uniquely ``physical'' frame.  Questions of observable
frame interpretation are logically separate from the field-redefinition
identity~\eqref{eq:null-positive}.

Finally, the analysis is classical.  Semiclassical stress tensors can violate
pointwise energy conditions, and classical versus quantum stability of Cauchy
horizons need not coincide~\cite{MarkovicPoisson1995}.  Quantum effects may
also destabilize horizons that are classically protected by degeneracy~\cite{McMaken2023}.  A
quantum extension would require a renormalized stress tensor and lies beyond
the present work.

\section{Discussion and conclusions}
\label{sec:conclusions}

We have developed a geometry-first classification of Cauchy-horizon mass
inflation in metric $f(R)$ gravity, based on the exact Hawking-mass transport
law~\eqref{eq:mass-transport-v}. At a regular nondegenerate Cauchy horizon,
the longitudinal source is exponentially blueshifted. Consequently, in the
eventually nonnegative and longitudinally dominated sector, a Price-type tail
drives mass inflation unless the effective source is suppressed strongly
enough to become blueshift integrable or another contribution to the exact
transport law competes at the same order. This distinction between genuine
bounded-mass screening and the weaker cancellation of only the leading
Price-tail contribution is central to the analysis. The corresponding
asymptotic regimes are illustrated in Fig.~\ref{fig:asymptotic-hierarchy},
while Fig.~\ref{fig:logic-overview} summarizes the local classification.

The leading-cancellation branch is itself highly constrained. The scalaron
cannot remain freely adjustable: its longitudinal evolution becomes slaved
to the decaying matter flux, while its transverse dynamics re-enters through
the trace equation. In the absence of an independent transverse or trace
cancellation, the curvature becomes unbounded and the gravitational
Lagrangian is driven toward an asymptotically linear branch. 
For regular model classes admitting a finite high-curvature limit of
$f_{RR}$, this further requires $f_{RR}\to0$.

The sign structure sharpens this result. On an eventually viable branch with
$f_{RR}>0$, the cancellation trajectory is forced toward
$R\to-\infty$ and must satisfy $\Lambda_\infty>-1$. The familiar
positive-curvature GR-recovery regime of viable cosmological $f(R)$ models
is therefore not automatically the branch sampled by the black-hole
interior. On the viable cancellation branch the standard scalaron mass
parameter becomes large, and under the additional rate control of
Proposition~\ref{prop:fRRrate} the conventional heavy, short-range scalaron
limit is recovered. Thus $f_{RR}\to0$ should not, by itself, be interpreted
as a pathology.

The Einstein-frame representation provides a complementary interpretation.
The scalaron's null kinetic contribution is nonnegative, so the Jordan-frame
cancellation of the matter source by derivatives of $F$ is not a cancellation
by negative scalaron kinetic energy. Rather, it is a redistribution between
scalar and geometric derivative terms, which helps explain why sustaining the
cancellation imposes strong asymptotic restrictions.

The analysis also separates scalaron screening from genuinely different
escape mechanisms. Horizon degeneracy removes the exponential blueshift at
its geometric origin; sufficiently rapid exterior decay can beat the
blueshift; mixed-channel competition can invalidate longitudinal dominance;
and transverse or trace cancellations can evade the rigidity argument.
Nonregular scalar-tensor limits, sign-changing conditionally convergent
sources, and violations of classical null positivity provide further
possibilities outside the regular sector considered here.

Several directions follow naturally. Numerically, existing double-null
$f(R)$ collapse codes can monitor the effective transport terms,
$\Lambda(v)$, $R$, and $f_{RR}$ to determine which asymptotic branch is
dynamically selected and whether the exceptional transverse cancellation can
arise from regular initial data. Geometrically, the same transport hierarchy
can be extended to slowly evolving inner trapping horizons, where large
interior amplification need not be tied to an exact global Cauchy horizon.

More broadly, it will be important to determine which elements of the rigidity
mechanism survive in genuine scalar-tensor, Horndeski, and higher-derivative
theories, whose inner-horizon dynamics must be derived from their full kinetic
structure. Semiclassical extensions are another natural direction, since
quantum stress tensors can violate the classical null-positivity assumptions
and horizon degeneracy need not imply quantum stability.

The main conclusion is therefore sharp. At a regular, nondegenerate Cauchy
horizon, metric $f(R)$ gravity does not provide a generic freely adjustable
scalaron screen for a Price-type influx. If the effective longitudinal source
remains blueshift nonintegrable, mass inflation follows; if its leading
contribution is cancelled, the scalaron is instead driven toward a restricted
rigidity branch unless an independent escape mechanism intervenes. Whether
such exceptional branches are dynamically formed and sufficiently regular to
alter the strong-cosmic-censorship picture is the central question left for
future investigation.

\appendix

\section{Derivation of the exact Hawking-mass transport law}
\label{app:mass}

For the metric~\eqref{eq:metric}, the relevant Einstein-tensor components are
\begin{align}
 G_{vv}
 &=-\frac{2}{r}\left(r_{,vv}-2\sigma_{,v}r_{,v}\right),
 \label{eq:Gvv-app}\\
 G_{uv}
 &=\frac{2r_{,uv}}{r}
 +\frac{2r_{,u}r_{,v}}{r^2}
 +\frac{\Omega^2}{2r^2}.
 \label{eq:Guv-app}
\end{align}
From $FG_{\mu\nu}=\mathcal P_{\mu\nu}$,
\begin{align}
 r_{,vv}-2\sigma_{,v}r_{,v}
 &=-\frac{r}{2F}\mathcal P_{vv},
 \label{eq:rvv-app}\\
 r_{,uv}
 &=\frac{r}{2F}\mathcal P_{uv}
 -\frac{r_{,u}r_{,v}}{r}
 -\frac{\Omega^2}{4r}.
 \label{eq:ruv-app}
\end{align}
The Hawking mass is
\begin{equation}
 m=\frac r2\left(1+\frac{4r_{,u}r_{,v}}{\Omega^2}\right).
 \label{eq:m-app}
\end{equation}
Differentiating with respect to $v$ gives
\begin{align}
 m_{,v}
 ={}&\frac{r_{,v}}2
 +\frac{2r_{,u}r_{,v}^2}{\Omega^2}
 \notag\\
 &+\frac{2r}{\Omega^2}
 \left[
 r_{,uv}r_{,v}
 +r_{,u}r_{,vv}
 -2\sigma_{,v}r_{,u}r_{,v}
 \right].
 \label{eq:mv-app1}
\end{align}
Substitution of Eqs.~\eqref{eq:rvv-app} and~\eqref{eq:ruv-app} cancels the
purely geometric terms and leaves
\begin{equation}
 m_{,v}
 =\frac{r^2}{F\Omega^2}
 \left(\mathcal P_{uv}r_{,v}-\mathcal P_{vv}r_{,u}\right),
 \end{equation}
which is Eq.~\eqref{eq:mass-transport-v}.  The $u$ equation follows by
interchanging $u$ and $v$.

\section{Asymptotic integral used in the screening lemma}
\label{app:asymptotic}

Let
\begin{equation}
 I(v)=\int_{v_0}^{v}\ee^{-2\sigma(\bar v)}\tau(\bar v)\,\dd\bar v.
 \label{eq:I-def}
\end{equation}

Under Eq.~\eqref{eq:CH-sigma}, namely
$\sigma_{,v}\to-\beta_-/2$, and $\tau'/\tau\to0$, define
\begin{equation}
 H(v)=\frac{\ee^{-2\sigma(v)}\tau(v)}{\beta_-}.
\end{equation}
Both $I$ and $H$ diverge because
$\ee^{-2\sigma}\tau\to\infty$. Moreover,
\begin{equation}
 \frac{I'}{H'}
 =\frac{\beta_-}
 {-2\sigma_{,v}+\tau'/\tau}
 \longrightarrow1.
\end{equation}
L'H\^opital's rule therefore gives $I/H\to1$, which proves
Eq.~\eqref{eq:asymptotic-integral} without invoking a specific power-law
index.

\section{Einstein-frame null positivity}
\label{app:EF}

For completeness, varying the matter action at fixed scalaron under
$g_{\mu\nu}=F^{-1}\widetilde g_{\mu\nu}$ gives
\begin{equation}
 \widetilde T^{(m)}_{\mu\nu}=F^{-1}T^{(m)}_{\mu\nu}.
\end{equation}
The canonical scalar stress tensor~\eqref{eq:phi-stress} obeys
\begin{equation}
 T^{(\phi)}_{\mu\nu}k^\mu k^\nu=(k^\mu\partial_\mu\phi)^2
\end{equation}
for every $\widetilde g$-null vector $k^\mu$.  Because conformal
transformations preserve null directions, the total null projection is
Eq.~\eqref{eq:null-positive}.  Thus the sign statement used in
Sec.~\ref{sec:einstein} is exact and independent of the scalar potential.

\begin{acknowledgments}
	FSNL acknowledges support from the Funda\c{c}\~{a}o para a Ci\^encia e a Tecnologia (FCT) Scientific Employment Stimulus contract with reference CEECINST/00032/2018, and funding through the research grants UIDB/04434/2020, UIDP/04434/2020 and PTDC/FIS-AST/0054/2021.
\end{acknowledgments}

\bibliographystyle{apsrev4-2}
\bibliography{Mass_Inflation_or_Scalaron_Rigidity_References}

\end{document}